\documentclass[11pt, onecolumn]{article}
\usepackage[top=1in, bottom=1in, left=1.25in, right=1.25in]{geometry}

\usepackage{amsfonts}
\usepackage{amsmath,amssymb}
\usepackage{graphicx,float}
\usepackage{color, soul}
\usepackage{algorithm,algorithmic}
\usepackage{bm}
\usepackage{booktabs}
\usepackage{flushend}
\usepackage{caption2}
\usepackage{tikz}
\usetikzlibrary{arrows}
\usepackage{subfigure}

\usepackage[amsmath,thmmarks]{ntheorem}
\usepackage{theorem}

\newtheorem{lem}{Lemma}
\newtheorem{defi}{Definition}
\newtheorem{rem}{Remark}
\newtheorem{thm}{Theorem}

\theoremheaderfont{\sc}\theorembodyfont{\upshape}
\theoremstyle{nonumberplain}
\theoremseparator{}
\theoremsymbol{\rule{1ex}{1ex}}
\newtheorem{proof}{Proof}

\newcommand{\vx}{{\bf x}}
\newcommand{\va}{{\bf a}}
\newcommand{\ve}{{\bf e}}

\newcommand{\myexp}[1]{\exp\left(#1\right)}

\begin{document}

\title{Linear Convergence of An Iterative Phase Retrieval Algorithm with Data Reuse}

\author{Gen~Li, Yuchen~Jiao, and Yuantao~Gu%
\thanks{The authors are with the Department of Electronic Engineering, Tsinghua University, Beijing, China. The corresponding author of this work is Y.~Gu (Email: gyt@tsinghua.edu.cn). }}

\date{submitted December 5, 2017}

\maketitle

\begin{abstract}
Phase retrieval has been an attractive but difficult problem rising from physical science,
and there has been a gap between state-of-the-art theoretical convergence analyses and the corresponding efficient retrieval methods.
Firstly, these analyses all assume that the sensing vectors and the iterative updates are independent, 
which only fits the ideal model with infinite measurements but not the reality, where data are limited and have to be reused.
Secondly, the empirical results of some efficient methods, such as the randomized Kaczmarz method, 
show linear convergence, which is beyond existing theoretical explanations considering its randomness and reuse of data.
In this work, we study for the first time, without the independence assumption, the convergence behavior of the randomized Kaczmarz method for phase retrieval.
Specifically, beginning from taking expectation of the squared estimation error with respect to the index of measurement by fixing the sensing vector and the error in the previous step,
we discard the independence assumption, rigorously derive the upper and lower bounds of the reduction of the mean squared error, 
and prove the linear convergence.
This work fills the gap between a fast converging algorithm and its theoretical understanding.
The proposed methodology may contribute to the study of other iterative algorithms for phase retrieval and other problems in the broad area of signal processing and machine learning.

\end{abstract}

\section{Introduction}


Phase retrieval is to recover a vector from some magnitude measurements,
which is equivalent to solving a system of the following quadratic equations,
\begin{align}\label{yr2}
y_r = \left|\langle \va_r, \vx^* \rangle\right|, \quad r = 1, 2, \ldots, m,
\end{align}
where $\vx^* \in \mathbb{C}^n$ is an unknown signal to be recovered,
$\va_r \in \mathbb{C}^n$, $y_r$, and $m$ denote the known $r$th sampling vector, the $r$th measurement,
and the total number of measurements, respectively.
In most previous works it is assumed that the sampling vectors are independent random variables following the distribution
$\mathcal{N}\left(0, \frac{1}{2}\mathbf{I}\right)+{\rm i}\mathcal{N}\left(0, \frac{1}{2}\mathbf{I}\right)$.
Apparently, $\vx^*{\rm e}^{{\rm i}\theta}$ is also a solution for any $\theta \in [0, 2\pi)$,
so the uniqueness of the solution to the phase retrieval problem can only be defined up to a global phase.
It has been shown that a unique solution can be determined if $m \ge (4n-4)$ \cite{balan2006signal, conca2015algebraic}.
For a real vector $\vx^* \in \mathbb{R}^n$, if $\va_r \in \mathbb{R}^n$ and $\va_r \sim \mathcal{N}(0, \mathbf{I})$ independently,
then $2n - 1$ measurements are sufficient.

The phase retrieval problem has appeared frequently in science and engineering, such as X-ray crystallography \cite{harrison1993phase},
microscopy \cite{miao2008extending}, astronomy \cite{fienup1987phase}, diffraction and array imaging \cite{bunk2007diffractive},
and optics \cite{walther1963question}. Other fields of application include acoustics, blind channel estimation in wireless communications,
interferometry, quantum mechanics, and quantum information \cite{shechtman2015phase}.
Focusing on the literature on these physical science fields, one can find the phase retrieval problem commonly encountered,
because most sensors in such areas can only record the intensity of some fields without the phase information.
Due to its wide applications, many algorithms have been proposed, and we will review them briefly in the following.

The classical algorithms for phase retrieval are the error reduction algorithm and its generalizations \cite{gerchberg1972practical, fienup1982phase}.
These algorithms alternate between the estimates of the missing phase and the unknown signals iteratively.
As suggested by their names, these algorithms satisfy the residual reduction property, and are often empirically shown to be effective,
but they lack rigorous theoretical performance guarantees.

Another popular method, PhaseLift, approaches the problem through reconstructing a rank-one matrix,
from which the unknown signal can be obtained \cite{candes2013phaselift, candes2014solving, waldspurger2015phase}.
The reconstruction can be solved using tractable semi-definite programming (SDP)-based convex relaxations.
PhaseLift is known to be able to provide exact solutions (up to a global phase) using
a near minimal number of sampling vectors \cite{balan2006signal}.
However, the computational complexity and memory requirement for
SDP-based algorithms become prohibitive as the dimension of the signal increases.

Recently, many iterative methods have arisen including the alternating minimization method \cite{netrapalli2013phase},
phase retrieval via Kaczmarz method \cite{Wei2015Solving, li2015phase}, and the Wirtinger Flow algorithm and its variants \cite{candes2015phase, chen2015solving}, which directly attack the phase retrieval problem in its original non-convex formulation.
In the random online setting, these iterative methods can achieve linear rate of convergence to a solution.
Moreover, \cite{li2015phase} establishes an exact analysis of the dynamics of the 
Kaczmarz method for phase retrieval in the large systems limit.

\subsection{Motivation}

In current theoretical works on iterative algorithms for phase retrieval,
the independence assumption has always been adopted to make the analysis mathematically easier.

\begin{defi}[Independence Assumption]
In an iterative algorithm for solving phase retrieval problem \eqref{yr2},
denote $\vx_{t-1}$ as the temporary estimate of $\vx^*$ before the $t$th iteration,
and $\va_t$ as the sensing vector used to measure $\vx^*$ and update the estimation in the $t$th iteration.
It is assumed that $\vx_{t-1}$ and $\va_t$ are independent.
\end{defi}

\begin{figure*}[t]
\begin{center}
\includegraphics[width=\textwidth]{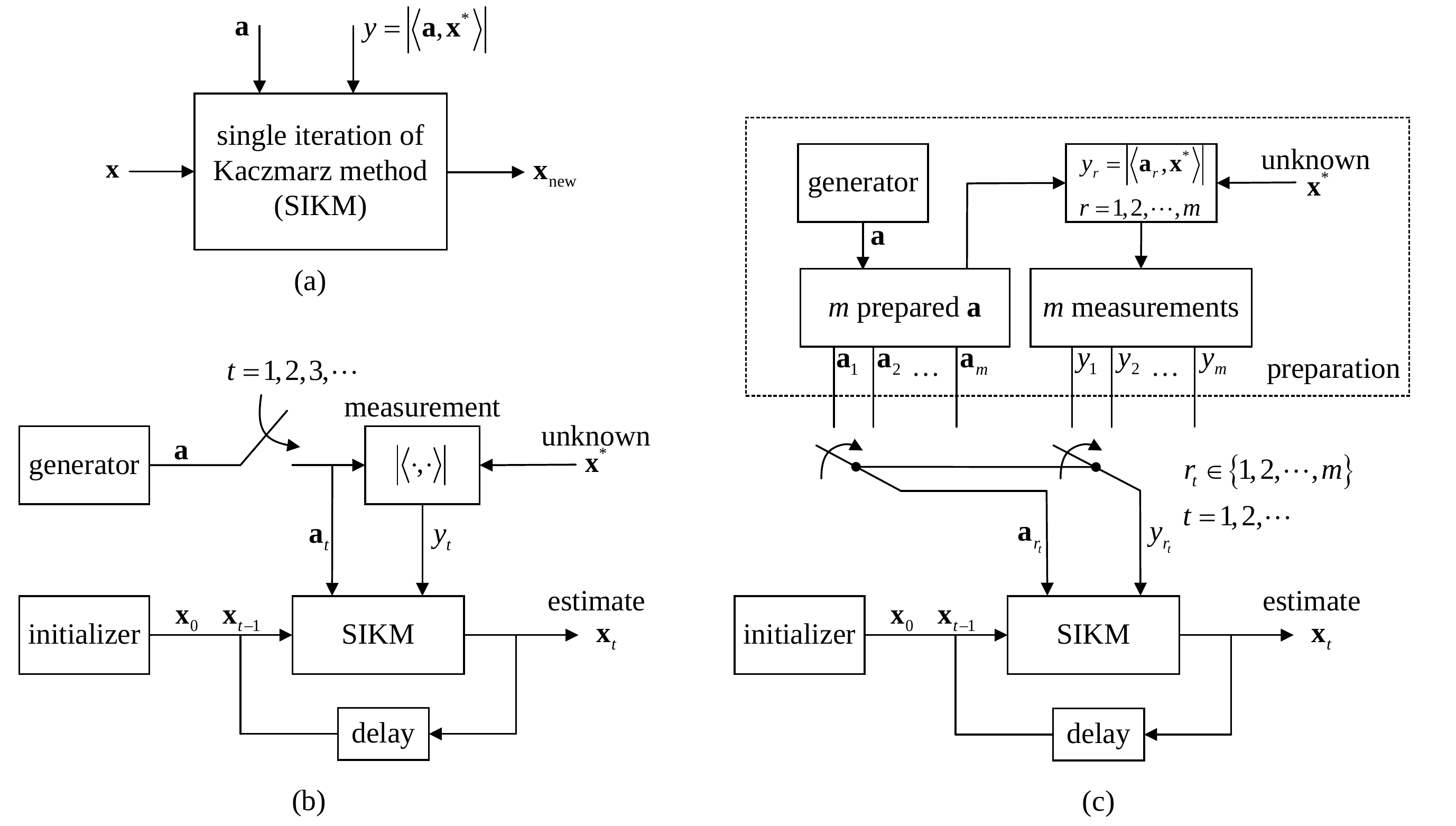}
\caption{\label{system}Visualization of phase retrieval via the randomized Kaczmarz method with both finite and infinite measurements.
(a) We introduce a functional block to denote the single iteration of Kaczmarz method (SIKM).
(b) The ideal case of infinite measurements or \emph{online processing},
where a new sensing vector is generated in each iteration and used to measure the unknown $\vx^*$.
Notice that $\va_t$ is independent with the temporary estimate $\vx_{t-1}$.
(c) The real case with finite measurements.
Before processing, $m$ sensing vectors are generated and used to produce $m$ measurements.
While in the $t$th iteration, a pair of sensing vector and measurement, denoted by $(\va_{r_t}, y_{r_t})$ are randomly chosen from the pool and sent to SIKM.
Notice that $\va_{r_t}$ is dependent with $\vx_{t-1}$ because the former may have already been used to produce the latter in previous iterations.
}
\end{center}
\end{figure*}

The independence assumption holds in the ideal case, where $m$ approaches infinity, or the scenario of \emph{online processing},
where the sensing vector is randomly generated for every measurement obtained for real-time processing.
However, in practical applications, it is not true.
The sensing vectors and the corresponding measurements have to be repeatedly used,
in that in order to reach a high precision estimate, the number of iterations of the algorithm is usually larger than
the number of measurements $m$.
In this case, the independence assumption does not hold,
for the reason that if $\va_r$ and $y_r$ have contributed to the estimate $\vx_{t-1}$ and are picked up again in the $t$th iteration,
then $\vx_{t-1}$ and $\va_r$ are dependent.
Please refer to Fig.~\ref{system}(b) and (c) for a visualization of the ideal online case and the real case.

Whether $\va_r$ and $\vx_{t-1}$ are independent leads to a visible difference in the analytical convergence behavior of the algorithm with finite and infinite measurements.
This will be visualized in Fig.~\ref{sim} in next section.
However, because it can alleviate difficulty in the analysis \cite{candes2015phase, chen2015solving, li2015phase, tan2017phase, jeong2017convergence},
the independence assumption is still widely adopted in the available works.
Such studies ignore the difference in the convergence performance displayed in Fig.~\ref{sim},
so we find their analyses not convincing enough.
Since the case with finite measurements is more common in reality, where the number of data is always limited,
more convincing theoretical analysis in the finite observation setting, without the independence assumption, 
is of importance in both theory and application,
which motivates our study.

We choose to analyze the randomized Kaczmarz method in phase retrieval.
According to the simulation results, this algorithm has good performance in both convergence rate and computational cost.
As a general row-action method,
its computational complexity is only $\mathcal{O}(n)$ per iteration
\cite{kaczmarz1937angenaherte, strohmer2006randomized}.

\subsection{Related Works}

\subsubsection{Phase retrieval using iterative projections: Dynamics in the large systems limit \cite{li2015phase}}

In our earlier work, we established an exact analysis of the dynamics of the algorithm in the large systems limit in the online setting.
The conclusion in that work will be introduced in detail in the next section.

\subsubsection{Phase retrieval via Wirtinger Flow: theory and algorithms \cite{candes2015phase}}

Phase Retrieval via Wirtinger Flow is shown to be able to achieve linear rate of convergence to a solution under the independence assumption between the iterative variable and the sampling vectors, i.e., Lemma 7.1, 7.2, 7.3 and so on.
When $m \ge cn\log n$, with high probability, the distance between the estimate of the algorithm at the $k$th iteration and the true signal $\vx^*$ decays exponentially (see Theorem 3.3).

\subsubsection{Solving random quadratic systems of equations is nearly as easy as solving linear systems \cite{chen2015solving}}

This paper modified Phase Retrieval via Wirtinger Flow by dropping terms bearing too much influence on the search direction to improve practical performance.
With similar argument as \cite{candes2015phase}, the authors prove that when $m \ge cn$, with high probability, the distance between the estimate of the algorithm at the $k$th iteration and the true signal $\vx^*$ decays exponentially (see Theorem 1).

\subsubsection{Phase Retrieval via randomized Kaczmarz: theoretical guarantees \cite{tan2017phase} and Convergence of the randomized Kaczmarz method for Phase Retrieval \cite{jeong2017convergence}}

These two papers provide theoretical bound for Phase Retrieval via randomized Kaczmarz method (see Theorem 1.2 \cite{tan2017phase} and Theorem 1.1 \cite{jeong2017convergence}) with the independence assumption between the iterative variable and the sampling vectors, i.e., Lemma 2.2 \cite{tan2017phase} and Section 2.1 \cite{jeong2017convergence}.

\subsection{Main Contribution}

In this paper, we study the convergence of the randomized Kaczmarz method for phase retrieval.
This is the first theoretical analysis on a phase retrieval algorithm without the independence assumption.
At first glance, an observation $y_r$ may contribute less innovation to the updated iterate, when it has been used for many times in the iterations. 
To some surprise, however, we successfully prove that the linear rate of convergence to a solution still holds in the finite measurements setting
with repetitively used observations.
Our approach begins from taking expectation to the squared estimation error with respect to \emph{the index of measurement} instead of the sensing vector, by \emph{fixing the sensing vector and the estimation error in the previous step} rather than taking an expectation.
This technique enables us to derive the convergence bounds no matter what dependence the temporary update and the sensing vector have, even if the former is a function of the latter.
As a consequence, the inappropriate independence assumption is successfully discarded for the first time,
which makes this work distinguish from all the previous theoretical results on this topic.
We believe that it will encourage more theoretical studies on the finite measurements setting for phase retrieval, and even other problems,
for that the analytical methodology we proposed may be adopted to analyze problems such as low-rank matrix recovery, adaptive filtering, and independent component analysis.

\subsection{Organization}

In Sections II, we review the randomized Kaczmarz method and demonstrate the gap between experimental results with data reuse and the theoretical prediction based on infinite measurements.
In Section III, we present our main contribution that proves the exponential convergence of the randomized Kaczmarz method for the first time without independence assumption. 
We derive the upper bound and the lower bound of the mean squared estimation error, and verify that both of them are of linear convergence rate.
Section IV concludes this work.
Section V collects the proofs of Lemmas and some probability inequalities.

\section{Preliminary}

In this work, we focus on analyzing phase retrieval in the real case,
where $\vx^* \in \mathbb{R}^n$, $\va_r \in \mathbb{R}^n$,
and $\va_r \sim \mathcal{N}(0, \mathbf{I})$ are independent for $r = 1,2,\cdots,m$.

\subsection{Phase Retrieval via randomized Kaczmarz method}

\begin{algorithm}[!t]
\caption{ Phase Retrieval via randomized Kaczmarz method for real case.\label{algorithm-1}}
\label{alg:Framwork}
\renewcommand{\algorithmicrequire}{\textbf{Input:}}
\begin{algorithmic}[1]

\REQUIRE ~~
     $\{(\va_r, y_r), r = 1, \ldots, m\}$, initialize $\vx_0$ using the spectral method, maximal iteration number $T$, $t=1$.

\ENSURE ~~
    $\vx_T$ as an estimate for $\vx^*$.

\WHILE{$t\le T$}
\STATE Choose $r$ randomly from $\{1, \ldots, m\}$ uniformly. \\
\STATE Update $\vx_t$ by using \eqref{iteration}.\\
\STATE $t\leftarrow t+1$.
\ENDWHILE
\end{algorithmic}
\end{algorithm}

The Phase Retrieval via randomized Kaczmarz method was proposed in \cite{Wei2015Solving} and analyzed in \cite{li2015phase}.
If we know the sign of $\langle \va_r, \vx^* \rangle$, according to the Kaczmarz method, $\vx_t$ is obtained by projecting $\vx_{t-1}$ onto the hyperplane determined by the linear equation $\langle \va_r, \vx \rangle = \langle \va_r, \vx^* \rangle$,
where the sign of $\va_r^{\rm T}\vx_{t-1}$ is used to estimate the unknown sign of $\va_r^{\rm T}\vx^*$.
Then the iteration becomes
\begin{align}\label{iteration}
\vx_t = \vx_{t-1}+\frac{y_r{\rm sgn}(\va_r^{\rm T}\vx_{t-1}) - \va_r^{\rm T}\vx_{t-1}}{\|\va_r\|^2}\va_r.
\end{align}
The algorithm is summarized in Algorithm \ref{alg:Framwork}.
In order to illustrate the dependence caused by data reuse in the finite measurements case,
the implementation of the randomized Kaczmarz method is visualized in Fig.~\ref{system}
for both the finite and the infinite measurements cases.

\subsection{Finite Measurements and Infinite Measurements}

The convergence behavior under the case with infinite measurements has been theoretically analyzed in our earlier work \cite{li2015phase}.
Let $d_k$ be the squared error between the estimate of the algorithm at the $k$th iteration and the true signal $\vx^*$.
Let $d(t) = d_{\lfloor tn\rfloor}$.
As $n$ tends to infinity, the random sample paths of $d(t)$ will converge to a continuous time function governed by the solutions of two deterministic, coupled ordinary differential equations (ODEs) (see Proposition 2, \cite{li2015phase}).

\begin{figure}
\begin{center}
\includegraphics[width=0.6\textwidth]{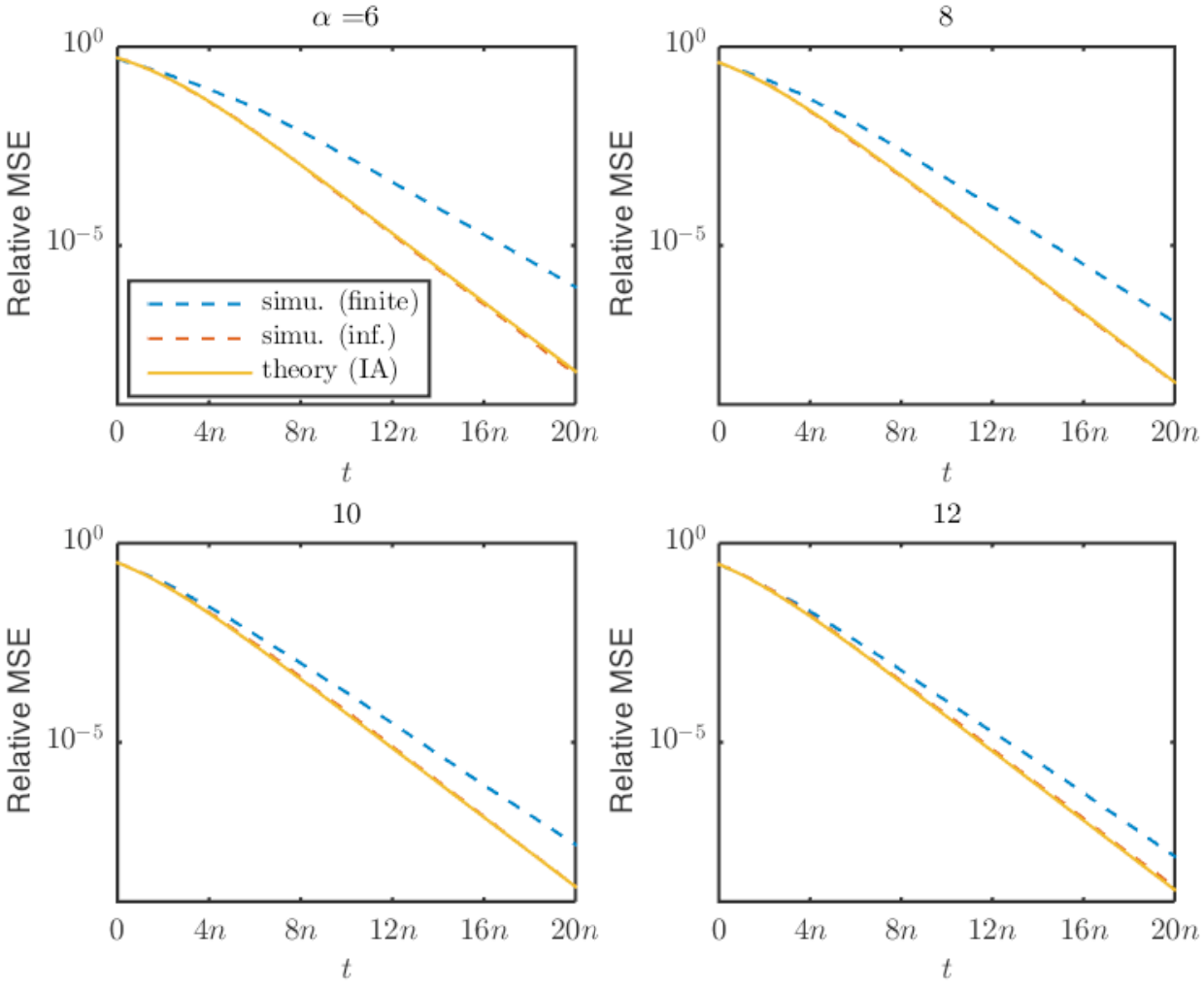}
\caption{\label{sim}Learning curves of the randomized Kaczmarz method in the numerical experiments and theoretical analysis with various $\alpha$.
$n$ is taken as 256.
The simulation results with finite and infinite measurements are denoted as ``simu.(finite)'' and ``simu.(inf.)'', respectively.
The theory with independence assumption is denoted by ``theory(IA)''.}
\end{center}
\end{figure}

Though our earlier theoretical result exactly predicts the learning curve with infinite measurements, 
it is far from applicable to the learning curve with finite measurements.
Define
$$
	\alpha := \frac{m}{n}
$$
as the sampling rate.
In Fig.~\ref{sim}, the theoretical result in \cite{li2015phase} and the experimental results with finite and infinite measurements of various sampling rate $\alpha$ are plotted in different colors.
We read that the theory based on the independence assumption in \cite{li2015phase} coincides with the simulation results in the infinite measurements case,
but when the dataset is limited and the algorithm reuses the data,
the convergence is slower than that using infinite data.
Moreover, even when we take $\alpha$ as a relatively large value such as $12$, 
the theoretical convergence speed is still noticeably faster than the experimental convergence speed in the
finite measurements case,
which suggests that even with a large $\alpha$
it is still inappropriate to use the independence assumption.
To get a better understanding of the algorithm with finite measurements, the abandon of the independence assumption in the analysis is in demand.

In addition, one may read from Fig.~\ref{sim} that even with data reuse the convergence of the randomized Kaczmarz method is still surprisingly fast, and the convergence rate appears to be linear.
Intuitively, the convergence speed may slow down when the iteration grows up,
because the \emph{impetus} or \emph{information} extracted from the measurements may gradually reduce to zero due to repeat use.
In this work, we successfully confirm the previously unjustified experimental convergence behavior by rigorous theoretical analysis which proves that the estimation error decreases exponentially during the  iterations.

%

\subsection{Initialization}

For algorithms with finite measurements, the initialization has great influence on the performance.
Some universal initialization methods with good performance have been proposed,
such as the spectral initialization originally introduced in \cite{netrapalli2013phase, candes2015phase} and
its generalization \cite{chen2015solving}.
Under the assumption that the sampling vectors consist of \emph{i.i.d.} Gaussian random variables, 
the result of the spectral initialization is aligned with the target vector $\vx^*$ in direction, 
when there are sufficiently many measurements \cite{chen2015solving}.
An exact high-dimensional analysis of the spectral method can be referred in \cite{lu2017phase}.

\section{Rigorous Analysis of randomized Kaczmarz Method without Independence Assumption}

The analysis of the convergence behavior of the randomized Kaczmarz method without the independence assumption is conducted in this section.
Though the data reuse damages the independence between the sensing vector and estimation error,
we surprisingly reveal in Theorem \ref{thm1} that this algorithm still has \emph{a linear convergence rate}.

\begin{thm}\label{thm1}
Let $\vx^* \in \mathbb{R}^n$ be any solution to the phase retrieval problem \eqref{yr2} in real case and
$\vx_t$ denote the $t$th iterative solution of Algorithm \ref{algorithm-1}.
Define
$$
{\rm dist}\left(\vx_t,\vx^*\right):=\min\left(\|\vx_t-\vx^*\|,\|\vx_t+\vx^*\|\right)
$$
to denote the estimate error up to a global phase.
Then there exist $\alpha_0$, $n_0$, when $n>n_0$, $\alpha>\alpha_0$, there exist constants $C_1>0$, $C_2>0$,
such that with probability at least
\begin{equation}\label{probability}
1-{\rm e}^{-C_1 n},
\end{equation}
we have
\begin{align}\label{bound}
\frac{{\mathbb{E}\rm dist}^2\left(\vx_t,\vx^*\right)}{{\rm dist}^2\left(\vx_{t-1},\vx^*\right)}\le 1-\frac{1}{n}C_2.
\end{align}
\end{thm}

\begin{rem}
It has been reviewed in the preliminary that there were practical methods providing good initialization with small estimation error.
Without loss of generality, we assume that
$$
\|\vx_t-\vx^*\|\le\|\vx_t+\vx^*\|
$$
always holds for all $t\ge0$.
Define
\begin{align}\label{def-e}
\ve_t:=\vx_t-\vx^*,
\end{align}
then \eqref{bound} can be derived from
\begin{align}\label{bound12}
\frac{\mathbb{E}\|\ve_t\|^2}{\|\ve_{t-1}\|^2}\le1-\frac{1}{n}C_2.
\end{align}
Throughout the proof below, we assume that the initialized $\vx_0$ is closer to $\vx^*$ than to $-\vx^*$.
Then we only need to prove \eqref{bound12}.
\end{rem}

\begin{rem}
We claim that there is no taking expectation to $\va_r$ in the LHS of \eqref{bound12}.
As far as we know, we propose for the first time this form to study the bound of the convergence behavior of an iterative phase retrieval algorithm.
The bound in \eqref{bound12} holds for arbitrary $\ve_{t-1}$, even if it is a function of the sensing vectors,
which is dependent with $\{\va_r, r = 1, \ldots, m\}$.
On the contrary, in existing works on such convergence analysis, it is always assumed that $\ve_{t-1}$ is a random vector independent of the sensing vectors, and then the expectation is taken,
although this assumption does not hold for the finite measurements case with data reuse.
\end{rem}

\begin{proof}
We prove Theorem \ref{thm1} in four steps.
In the first step, the iterative expression of squared estimation error is derived and taken expectation with respect to \emph{the index of data sample.}
This plays the essential role in discarding the independence assumption.
In the second step, we derive the bound of the mean squared estimation error and build its connection with the eigenvalues of random matrix.
In the consequent step, based on the concentration properties and the knowledge on random matrix,
we estimate the squared error.
Finally, the result is reshaped to a ready formulation to complete the proof.

{\bf Step 1)}
We will start from \eqref{iteration} and derive the iterative expression of $\ve_t$ first.
Substituting the definition of $y_r$ in \eqref{yr2} into \eqref{iteration} and replacing the absolute value by sign function,
we have
\begin{align}
\vx_t &=\vx_{t-1}+\frac{\left|\va_r^{\rm T}\vx^{*}\right|{\rm sgn}\left(\va_r^{\rm T}\vx_{t-1}\right) - \va_r^{\rm T}\vx_{t-1}}{\|\va_r\|^2}\va_r \nonumber\\ 
&= \vx_{t-1}+\frac{\va_{r}^{\rm T}\vx^*{\rm sgn}\left(\va_{r}^{\rm T}\vx^*\va_{r}^{\rm T}\vx_{t-1}\right)}{\|\va_r\|^2}\va_{r} - \frac{\va_{r}^{\rm T}\vx_{t-1}}{\|\va_r\|^2}\va_{r} \label{eq-temp4} \\
&= \vx_{t-1}+\frac{\va_{r}^{\rm T}\vx^*b_{r,t-1}}{\|\va_r\|^2}\va_{r}
 - \frac{\va_{r}^{\rm T}\vx_{t-1}-\va_{r}^{\rm T}\vx^*}{\|\va_r\|^2}\va_{r}, \label{vx_t}
\end{align}
where
$$
b_{r,t-1}:={\rm sgn}\left(\va_{r}^{\rm T}\vx^*\va_{r}^{\rm T}\vx_{t-1}\right)-1
$$
denotes whether the temporary solution at $t-1$ is \emph{wrong} with respect to the sign measured by $\va_r$,
and \eqref{vx_t} is obtained by inserting and removing $\frac{\va_r^{\rm T}\vx^*}{\|\va_r\|^2}\va_r$ into (from) \eqref{eq-temp4}.
Subtracting the ground truth $\vx^*$ from both sides of \eqref{vx_t} and
recalling the definition of $\ve_{t}$ in \eqref{def-e},
we have
\begin{align}
        \ve_t = & \ve_{t-1}+\frac{\va_r^{\rm T}\vx^*b_{r,t-1}}{\|\va_r\|^2}\va_{r}
         -\frac{\va_r^{\rm T}\ve_{t-1}}{\|\va_r\|^2}\va_{r}\nonumber\\
        = & \left({\bf I}-\frac{\va_r\va_r^{\rm T}}{\|\va_r\|^2}\right)\ve_{t-1}
        +\frac{\va_r^{\rm T}\vx^*b_{r,t-1}}{\|\va_r\|^2}\va_{r}\label{eq-temp6}.
\end{align}
Next we study the behavior of the squared estimation error.
According to \eqref{eq-temp6}, we write
\begin{equation}\label{eq-temp7}
        \|\ve_t\|^2 = \left\|\left({\bf I}-\frac{\va_r\va_r^{\rm T}}{\|\va_r\|^2}\right)\ve_{t-1}
        +\frac{\va_r^{\rm T}\vx^*b_{r,t-1}}{\|\va_r\|^2}\va_{r}\right\|^2.
\end{equation}
Notice that $\left({\bf I}-{\va_r\va_r^{\rm T}}/{\|\va_r\|^2}\right)$ is the projection matrix to a hyperplane which is perpendicular to $\va_r$.
Therefore the cross item in the RHS of \eqref{eq-temp7} must be zero and we get
\begin{equation}\label{eq-temp8}
        \|\ve_t\|^2 = \left\|\left({\bf I}-\frac{\va_r\va_r^{\rm T}}{\|\va_r\|^2}\right)\ve_{t-1}\right\|^2
        +\left\|\frac{\va_r^{\rm T}\vx^*b_{r,t-1}}{\|\va_r\|^2}\va_{r}\right\|^2.
\end{equation}
Using the perpendicular properties again, we may simplify the first item in the RHS of \eqref{eq-temp8} as
\begin{align}\label{eq-temp9}
        \left\|\left({\bf I}-\frac{\va_r\va_r^{\rm T}}{\|\va_r\|^2}\right)\ve_{t-1}\right\|^2
        &= \|\ve_{t-1}\|^2-\left\|\frac{\va_r\va_r^{\rm T}}{\|\va_r\|^2}\ve_{t-1}\right\|^2 \nonumber\\
        &= \|\ve_{t-1}\|^2-\left(\frac{\va_r^{\rm T}\ve_{t-1}}{\|\va_r\|}\right)^2.
\end{align}
Inserting \eqref{eq-temp9} into \eqref{eq-temp8}, we get the iteration of squared estimation error as
\begin{equation}\label{e-norm}
    \left\|\ve_t\right\|^2=\|\ve_{t-1}\|^2-\left(\frac{\va_r^{\rm T}\ve_{t-1}}{\|\va_r\|}\right)^2
    +\left(\frac{\va_r^{\rm T}\vx^*}{\|\va_r\|}\right)^2b_{r,t-1}^2.
\end{equation}
In order to get the average performance of one iteration from a fixed $\ve_{t-1}$, we take expectation with respect to $r$ on both sides of \eqref{e-norm}.
Considering that $r$ is a random variable uniformly distributed over $1,\cdots, m$, we get
\begin{align}
\mathbb{E}\left\|\ve_t\right\|^2 
=&\|\ve_{t-1}\|^2-\frac{1}{m}\sum_{k=1}^m\left(\frac{\va_k^{\rm T}\ve_{t-1}}{\|\va_k\|}\right)^2
+\frac{1}{m}\sum_{k=1}^{m}\left(\frac{\va_k^{\rm T}\vx^*}{\|\va_k\|}\right)^2b_{k,t-1}^2\label{Eve-step0}\\
=&\|\ve_{t-1}\|^2-\frac{1}{m}\sum_{k=1}^m\left(\frac{\va_k^{\rm T}\ve_{t-1}}{\|\va_k\|}\right)^2
+\frac{4}{m}\sum_{k=1}^{m}\left(\frac{\va_k^{\rm T}\vx^*}{\|\va_k\|}\right)^2\mathbb{I}_{k\in\mathcal{S}},\label{Eve-step1}
\end{align}
where $\mathbb{I}(\cdot)$ is the indicator function and $\mathcal{S}$ is defined as
\begin{align*}
\mathcal{S} & := \{k: \ b_{k,t-1}\ne 0\}\\
& := \{k: \ {\rm sgn}\left(\va_{k}^{\rm T}\vx^*\right)\neq{\rm sgn}\left(\va_k^{\rm T}\vx_{t-1}\right)\}.
\end{align*}
Equation \eqref{Eve-step1} comes from the fact that those items in the second summation in the RHS of \eqref{Eve-step0} are non-zero, if and only if $k\in\mathcal{S}$.
We then move those items with $k\in\mathcal{S}$ from the first summation to the second one and obtain
\begin{align}\label{Eve-step2}
\mathbb{E}\left\|\ve_t\right\|^2=\|\ve_{t-1}\|^2-\frac{1}{m}\sum_{k\in\bar{\mathcal{S}}}\left(\frac{\va_k^{\rm T}\ve_{t-1}}{\|\va_k\|}\right)^2
+\frac{1}{m}\sum_{k\in\mathcal{S}}\left(4\left(\frac{\va_k^{\rm T}\vx^*}{\|\va_k\|}\right)^2-\left(\frac{\va_k^{\rm T}\ve_{t-1}}{\|\va_k\|}\right)^2\right),
\end{align}
where $\bar{\mathcal{S}}$ denotes the  complement of $\mathcal S$.
Now we are ready to study the bound of the second and the third items in the RHS of \eqref{Eve-step2} by using the eigenvalues of random matrices.

We must stress that one cannot simply take expectation of \eqref{e-norm} with respect to $\va_r$ by fixing $\ve_{t-1}$.
The reason is that $\va_r$ is dependent on $\ve_{t-1}$ and is not a Gaussian random vector any more conditioning on $\ve_{t-1}$.
This is the key point of our approach of abandoning the independence assumption.

{\bf Step 2)}
Denote
\begin{align}
\beta  := &\frac{\left|\mathcal{S}\right|}{n},\label{eq-def-beta}\\
\alpha-\beta  = &\frac{\left|\bar{\mathcal{S}}\right|}{n},\nonumber
\end{align}
respectively, as the ratio of \emph{wrong} and \emph{correct} measurements number to unknown variable number,
where $\left|\cdot\right|$ denotes the cardinality of a set.

For the second item in the RHS of \eqref{Eve-step2},
we first write it in a matrix form and find its relationship with the eigenvalues of this matrix.
\begin{align}
\frac{1}{m}\sum_{k\in\bar{\mathcal{S}}}\frac{\left(\va_k^{\rm T}\ve_{t-1}\right)^2}{\|\va_k\|^2}& \geq \frac{1}{z_{\max}^2m}\sum_{k\in\bar{\mathcal{S}}}\left(\va_k^{\rm T}\ve_{t-1}\right)^2\nonumber\\
& = \frac{1}{z_{\max}^2}\frac{\alpha-\beta}{\alpha}\ve_{t-1}^{\rm T}\bm\Sigma_{\bar{\mathcal{S}}}\ve_{t-1},\label{second_term_pre1}
\end{align}
where
\begin{align*}
z_{\max}&=\max{\left\|\va_{k}\right\|},\\
\bm\Sigma_{\bar{\mathcal{S}}} &= \frac{1}{|\bar{\mathcal{S}}|}\sum_{k\in\bar{\mathcal{S}}}\va_k\va_k^{\rm T}.
\end{align*}

For the third item in the RHS of \eqref{Eve-step2}, we will bound $\left|\va_k^{\rm T}\vx^*\right|$ with $\left|\va_k^{\rm T}\ve_{t-1}\right|$ first, and then write it in the matrix form.
Notice that
\begin{align}
\va_k^{\rm T}\vx^*\va_k^{\rm T}\vx_{t-1} = \va_k^{\rm T}\vx^*\left(\va_k^{\rm T}\vx^*+\va_k^{\rm T}\ve_{t-1}\right).
\end{align}
For $k\in\mathcal{S}$, we know that the absolute value of $\va_k^{\rm T}\ve_{t-1}$ is large enough to change the sign of $\va_k^{\rm T}\vx_{t-1}$ different from $\va_k^{\rm T}\vx^*$.
As a consequence,
\begin{align}\label{cond-S}
\mathcal{S}\subset \left\{k:\ \left|\va_k^{\rm T}\vx^*\right|\le\left|\va_k^{\rm T}\ve_{t-1}\right|\right\},
\end{align}
and we have
\begin{align}\label{y-aTe}
\left(\frac{\va_k^{\rm T}\vx^*}{\|\va_k\|}\right)^2\le\left(\frac{\va_k^{\rm T}\ve_{t-1}}{\|\va_k\|}\right)^2, \quad k\in\mathcal{S}.
\end{align}
Substituting \eqref{y-aTe} into the third item in the RHS of \eqref{Eve-step2} and formulating the result in a matrix form, we get
\begin{align}
\frac{1}{m}\sum_{k\in\mathcal{S}}\left(4\left(\frac{\va_k^{\rm T}\vx^*}{\|\va_k\|}\right)^2-\left(\frac{\va_k^{\rm T}\ve_{t-1}}{\|\va_k\|}\right)^2\right)
\le& \frac{1}{z_{\min}^2}\frac{3}{m}\sum_{k\in \mathcal{S}}\left(\va_k^{\rm T}\ve_{t-1}\right)^2\nonumber \\
=&\frac{1}{z_{\min}^2}\frac{3\beta}{\alpha}\ve_{t-1}^{\rm T}\bm\Sigma_{\mathcal{S}}\ve_{t-1}\label{},
\label{third_term_pre1}
\end{align}
where
\begin{align}
z_{\min}&=\min{\left\|\va_{k}\right\|}\nonumber,\\
\bm\Sigma_{\mathcal{S}}&=\frac{1}{|\mathcal{S}|}\sum_{k \in\mathcal{S}}\va_k\va_k^{\rm T}.\label{def-Sigma}
\end{align}

Inserting \eqref{second_term_pre1} and \eqref{third_term_pre1} into \eqref{Eve-step2}
and utilizing the properties of eigenvalues,
$$
	\lambda_{\min}({\bf A})\|{\bf u}\|^2 \le {\bf u}^{\rm T}{\bf A}{\bf u} \le \lambda_{\max}({\bf A})\|{\bf u}\|^2,
$$
where $\lambda_{\min}(\cdot)$ and $\lambda_{\max}(\cdot)$ denotes, respectively, the smallest and largest eigenvalue of a matrix,
we arrive at
\begin{equation}\label{full_bound}
\frac{\mathbb{E}\|\ve_t\|^2}{\|\ve_{t-1}\|^2} \le 1 - \frac{\alpha-\beta}{z_{\max}^2\alpha}\lambda_{\min}\left(\bm\Sigma_{\bar{\mathcal{S}}}\right)
+ \frac{3\beta}{z_{\min}^2\alpha}\lambda_{\max}\left(\bm\Sigma_{\mathcal{S}}\right).
\end{equation}

Notice again that the bound in \eqref{full_bound} holds no matter what dependence $\ve_{t-1}$ and $\va_r$ have.
The inappropriate independence assumption is successfully abandoned in this work.
This is different from all the previous theoretical works on this topic.

{\bf Step 3)}
Now we need to estimate $z_{\max}, z_{\min}$, and the smallest and the largest eigenvalue of $\bm\Sigma_{\bar{\mathcal{S}}}$ and $\bm\Sigma_{\mathcal{S}}$, respectively.
According to their definitions, the estimation problems are closely related to the properties of Gaussian random matrix.
We then develop the following two lemmas, which could be readily used to solve the above estimation problems.

\begin{lem}\label{lem2}
Let $\va$ be a vector in $\mathbb{R}^{n}$, whose entries are independent standard Gaussian random variables. Then for any $0<\varepsilon < 1$, with probability at least
\begin{align*}
P_1(\varepsilon,n):=1-2\myexp{-n\left(\frac{\varepsilon^2}{4}-\frac{\varepsilon^3}{6}\right)},
\end{align*}
we have
\begin{align}\label{lem-eq1}
\left|\frac{\|\va\|^2}{n}-1\right|< \varepsilon.
\end{align}
\end{lem}
\begin{proof}
The proof is postponed to Appendix \ref{proof-lem2}.
\end{proof}

\begin{lem}\label{lem3}
Let $\va_1,\cdots,\va_m$ be $m$ vectors in $\mathbb{R}^{n}$, whose entries are independent standard Gaussian random variables.
For given integer $1\le p \le m$ and any $\mathcal{S}\subset \{1,\cdots,m\}$ satisfying $|\mathcal{S}|=p$,
we can use $\{\va_i\}_{i \in \mathcal{S}}$ to compose an $n\times n$ matrix
$$
{\bm\Sigma}_{\mathcal{S}}=\frac{1}{p}\sum_{i\in \mathcal{S}}\va_i\va_i^{\rm T}.
$$
Then with probability at least
\begin{align}\label{P3}
P_2\left(\varepsilon_1,p\right):=1- C_{m}^{p}\myexp{-\frac{p\varepsilon_1^2}{2}},
\end{align}
we have
\begin{align}\label{lem-eq3}
\max_{\mathcal{S}:|\mathcal{S}|=p}\lambda_{\max}\left({\bm\Sigma}_{\mathcal{S}}\right) \le \left(1 + \frac{1}{\sqrt{p/n}} + \varepsilon_1\right)^2.
\end{align}
If $p\ge n$, with probability at least $P_2\left(\varepsilon_2,p\right)$,
we have
\begin{align}\label{lem-eq2}
\min_{\mathcal{S}:|\mathcal{S}|=p}\lambda_{\min}\left({\bm\Sigma}_{\mathcal{S}}\right) \ge \left(1 - \frac{1}{\sqrt{p/n}} - \varepsilon_2\right)^2.
\end{align}
\end{lem}
\begin{proof}
The proof is postponed to Appendix \ref{proof-lem3}.
\end{proof}


Now we use these Lemmas to estimate the last two items in the RHS of \eqref{full_bound}.
Specifically, with probability at least
\begin{align}\label{P_pre_final}
P(\beta) = &1-m\left(1-P_1(\varepsilon_1,n)\right) -\left(1-P_2\left(\varepsilon_2,\beta n\right)\right)-\left(1-P_2\left(\varepsilon_3,m-\beta n\right)\right)\nonumber\\
 = &1-2m\myexp{-n\left(\frac{\varepsilon_1^2}{4}-\frac{\varepsilon_1^3}{6}\right)}
 -{\rm C}_{m}^{\beta n}\myexp{-\frac{\beta n\varepsilon_2^2}{2}}
 -{\rm C}_{m}^{(\alpha -\beta)n}\myexp{-\frac{\left(\alpha -\beta\right)n\varepsilon_3^2}{2}},
\end{align}
we have
\begin{align}\label{Eve_pre_final}
\frac{\mathbb{E}\left\|\ve_t\right\|^2}{\|\ve_{t-1}\|^2}
\le 1\!-\!\left(\!1\!-\frac{\beta}{\alpha}\right)\frac{1}{n(1+\varepsilon_1)}\left(1\!-\!\frac{1}{\sqrt{\alpha\!-\!\beta}}-\varepsilon_3\right)^2
+\frac{3\beta}{\alpha}\frac{1}{n\left(1-\varepsilon_1\right)}\left(1+\frac{1}{\sqrt{\beta}}+\varepsilon_2\right)^2.
\end{align}

{\bf Step 4)}
To complete the proof, we need to formulate \eqref{P_pre_final} and \eqref{Eve_pre_final} into the shape of \eqref{probability} and \eqref{bound12}, respectively, and demonstrate that both $C_1$ and $C_2$ are positive.

In order to remove the influence of $\beta$, we first calculate its upper bound, denoted by $\beta_0$,
and then study the lower bound of convergence rate. 
According to our analysis, where the detail is postponed to Appendix \ref{appendix-beta0},
when $n$ approaches infinity with $\alpha$ fixed, the asymptotical bound $\beta'_0$ is determined by
\begin{align}\label{eq-solving-beta0}
\left(1 - \frac{\alpha}{\beta'_0}\frac{2\tau}{\sqrt{2\pi}}\myexp{-\frac{\tau^2}{2}}\right)\|\vx^*\|^2 =
 \left(1 + \frac{1}{\sqrt{\beta'_0}} + \sqrt{2\ln \frac{{\rm e}\alpha}{\beta'_0}}\right)^2\|\ve_{t-1}\|^2,
\end{align}
where $1-2Q\left(\tau\right)={\beta'_0}/{\alpha}$ and $Q\left(\cdot\right)$ is the tail probability of the standard Gaussian distribution.
Because of the continuity,
for all $\delta_\beta>0$, there exists $n_\beta >0$,
such that for all $n\ge n_\beta$, we can use
$
	\beta_0 := \beta'_0 + \delta_\beta
$
as a bound.

By introducing
\begin{align}
{\hat\varepsilon}_1& :=\varepsilon_1,\quad\label{def-hat-e-1}\\
{\hat\varepsilon}_2& :=\sqrt{\varepsilon_2^2-2\ln\frac{{\rm e}\alpha}{\beta_0}},\quad\label{def-hat-e-2}\\
{\hat\varepsilon}_3& :=\sqrt{\varepsilon_3^2-2\frac{\beta_0}{\alpha-\beta_0}\ln\frac{{\rm e}\alpha}{\beta_0}},\label{def-hat-e-3}
\end{align}
we could formulate $P(\beta_0)$ defined in \eqref{P_pre_final} into the shape of \eqref{probability} and verify $C_1>0$ for all $n$ larger than a constant $n_{C_1}$.
The details are necessary but tedious, therefore they are postponed to Appendix \ref{proof-C1p}.

Replacing $\beta$ by $\beta_0$ and substituting \eqref{def-hat-e-1}, \eqref{def-hat-e-2}, and \eqref{def-hat-e-3}  into \eqref{Eve_pre_final}, we claim that with probability at least $P(\beta_0)$
the following holds
\begin{equation}\label{Eve_mid}
\frac{\mathbb{E}\|\ve_{t}\|^2}{\|\ve_{t-1}\|^2} \le 1 - \frac{1}{n}C_2(n,\alpha),
\end{equation}
where
\begin{align*}
C_2(n,\alpha) :=
&\left(1-\frac{\beta_0}{\alpha}\right)\frac{1}{1+\hat{\varepsilon}_1}\left(1-\frac{1}{\sqrt{\alpha\!-\!\beta_0}}-\sqrt{\frac{2\beta_0}{\alpha\!-\!\beta_0}\ln\frac{{\rm e}\alpha}{\beta_0}+\hat{\varepsilon}_3^2}\right)^2\\
&-\frac{3\beta_0}{\alpha}\frac{1}{1-\hat{\varepsilon}_1}\left(1+\frac{1}{\sqrt{\beta_0}}+\sqrt{2\ln\frac{{\rm e}\alpha}{\beta_0}+\hat{\varepsilon}_2^2}\right)^2.
\end{align*}
Our final task is to analyze that $C_2$ is positive when $n$ is larger than a constant $n_0$ and $\alpha$ is larger than a constant $\alpha_0$.
We will consider the asymptotical case of $n$ approaching infinity to verify this issue.
If we could prove that $C_2(n,\alpha)$ tends towards a positive constant when $n$ goes to infinity and $\alpha$ is large enough,
then there always exist $n_{C_2}$, $\alpha_0$, when $n>n_{C_2}$, $\alpha>\alpha_0$,  we have $C_2 >0$.
As a consequence, we may take $n_0=\max(n_\beta,n_{C_1},n_{C_2})$ and close the proof.

According to Lemma \ref{lem2}, Lemma \ref{lem3} and the derivation in Appendix \ref{proof-C1p},
when $n$ goes to infinity, let $\hat{\varepsilon}_i\to 0, i=1,2,3$, we have
\begin{align}
\lim_{n\to\infty}\frac{\mathbb{E}\|\ve_{t}\|^2-\|\ve_{t-1}\|^2}{\|\ve_{t-1}\|^2/n}
\le&  \lim_{n\to\infty} -C_2(n,\alpha)\nonumber\\
 =&-\frac{\alpha-\beta_0}{\alpha}\left(1-\frac{1}{\sqrt{\alpha-\beta_0}}-\sqrt{\frac{2\beta_0}{\alpha-\beta_0}\ln\frac{{\rm e}\alpha}{\beta_0}}\right)^2\nonumber\\
&+\frac{3\beta_0}{\alpha}\left(1+\frac{1}{\sqrt{\beta_0}}+\sqrt{2\ln\frac{{\rm e}\alpha}{\beta_0}}\right)^2.\label{Eve_final}
\end{align}

According to \eqref{eq-solving-beta0}, when $\|\ve_{t-1}\|$ is small enough, $\beta_0$ is nearly zero.
In this case, \eqref{Eve_final} is approximated by
\begin{equation}\label{boundlim}
\lim_{n\to\infty}C_2(n,\alpha) \approx -\frac{3}{\alpha}+ \left(1 - \frac{1}{\sqrt{\alpha}}\right)^2.
\end{equation}
We read from $\eqref{boundlim}$ that there exits some $\alpha_0$, when $\alpha>\alpha_0$,
the RHS of \eqref{boundlim} is positive.
Then when the initialization is good enough
\footnote{A good initialization is possible through some practical methods, such as the truncated spectral method \cite{chen2015solving}.} so that $\|\ve_{t-1}\|$ is small enough, $C_2$ approaches to a positive constant as $n$ goes to infinity.
Based on our previous discussion, this completes the proof well, and we demonstrate that algorithm \ref{algorithm-1} will converge to a solution linearly.
\end{proof}

\begin{rem}\label{rem_lowerbound1}
We have found the upper bound of the relative mean squared estimation error without the independence assumption.
In order to better understand the convergence behavior of this algorithm, we also need to study its lower bound.
This will be accomplished based on the iteration of \eqref{Eve-step1},
of which the third item is always no less than zero.
For the second item, we can find its upper bound by using Lemma \ref{lem2} and \ref{lem3}.
That is
\begin{align*}
\frac{1}{m}\sum_{k=1}^m\left(\frac{\va_k^{\rm T}\ve_{t-1}}{\|\va_k\|}\right)^2\le& \frac{1}{z_{\min}^2}\ve_{t-1}^{\rm T}{\bm\Sigma}\ve_{t-1}\\
\le& \frac{1}{z_{\min}^2}\lambda_{\max}\left({\bm\Sigma}\right)\|\ve_{t-1}\|^2,
\end{align*}
where
$$
{\bm\Sigma}:=\frac{1}{m}\sum_{k=1}^m\va_k\va_k^{\rm T}.
$$
As a consequence we write the lower bound as
\begin{equation}\label{eq-lower-bound1}
\lim_{n \to \infty}\frac{\mathbb{E}\|\ve_{t}\|^2-\|\ve_{t-1}\|^2}{\|\ve_{t-1}\|^2/n} \geq -\left(1 + \frac{1}{\sqrt{\alpha}}\right)^2.
\end{equation}
\end{rem}

\section{Conclusion}
In this paper, we reveal the linear rate of convergence of Phase Retrieval via the randomized Kaczmarz method in the real case with finite measurements.
The main advantage over the previous theoretical approaches for phase retrieval
is that the independence assumption inappropriate for the date reuse case
is discarded in our analysis.
The upper and lower bounds of the rate of convergence are given. 
The methodology we used could be adopted to analyze other problems.

\section{Appendix}

\subsection{Proof of Lemma \ref{lem2}}
\label{proof-lem2}

The proof will be completed, if we could verify the two inequalities below
\begin{align}
\mathbb{P}\left(\frac{\|\va\|^2}{n}\geq 1+\varepsilon\right)&\le\myexp{-n\left(\frac{\varepsilon^2}{4}-\frac{\varepsilon^3}{6}\right)},\label{subconclusion1}\\
\mathbb{P}\left(\frac{\|\va\|^2}{n}\leq 1-\varepsilon\right)&\le\myexp{-n\left(\frac{\varepsilon^2}{4}-\frac{\varepsilon^3}{6}\right)}.\label{subconclusion2}
\end{align}
Let's begin from the LHS of \eqref{subconclusion1} and write
\begin{align}
\mathbb{P}\left(\frac{\|\va\|^2}{n}\geq 1+\varepsilon\right)&=\mathbb{P}\left(\myexp{\lambda\frac{\|\va\|^2}{n}}\ge \myexp{\lambda\left(1+\varepsilon\right)}\right)\nonumber\\
&\leq\frac{\mathbb{E}\left(\myexp{\lambda\frac{\|\va\|^2}{n}}\right)}{\myexp{\lambda\left(1+\varepsilon\right)}},\label{lemma1_temp1}
\end{align}
where \eqref{lemma1_temp1} comes from the Markov's inequality for $X > 0, a > 0$ that
$$
\mathbb{P}\left(X\ge a\right)\le \frac{\mathbb{E}\left(X\right)}{a}
$$
and $\lambda$ is a positive parameter to increase the degrees of freedom to get a more tight bound.

Notice that $\|{\bf a}\|^2$ follows the $n$-dimensional Chi-square distribution.
Adopting the characteristic function of Chi-square distributed variable $x$
\begin{align}\label{Chrc-fun}
\mathbb{E}\left({\rm e}^{{\rm i}tx}\right)=\left(1-2{\rm i}t\right)^{-\frac{n}{2}},
\end{align}
and let $t=-{\rm i}\frac{\lambda}{n}$,
we derive from \eqref{lemma1_temp1}
\begin{align}
\mathbb{P}\left(\frac{\|\va\|^2}{n}\geq 1+\varepsilon\right)
\le&\left(1-2\frac{\lambda}{n}\right)^{-\frac{n}{2}}\myexp{-\lambda\left(1+\varepsilon\right)}\nonumber\\
=&\myexp{-\lambda\left(1+\varepsilon\right)-\frac{n}{2}\ln\left(1-2\frac{\lambda}{n}\right)}\nonumber\\ 
\le& \max_{\lambda} \myexp{-\lambda\left(1+\varepsilon\right)-\frac{n}{2}\ln\left(1-2\frac{\lambda}{n}\right)}\nonumber\\
 =& \myexp{\frac{n}{2}\left(-\varepsilon+\ln\left(1+\varepsilon\right)\right)}.\label{P21-ori}
\end{align}

By approximating the natural logarithm in \eqref{P21-ori} as its Taylor's series of $\varepsilon$,
we can get
\begin{align}\label{series1}
\frac{n}{2}\left(-\varepsilon+\ln\left(1+\varepsilon\right)\right)
=&\frac{n}{2}\left(-\varepsilon+\left(\varepsilon-\frac{1}{2}\varepsilon^2+\frac{1}{3}\varepsilon^3+\sum_{k=4}^{\infty}\frac{\left(-1\right)^{k+1}}{k}\varepsilon^k\right)\right)\nonumber\\
<&-n\left(\frac{\varepsilon^2}{4}-\frac{\varepsilon^3}{6}\right).
\end{align}
Inserting \eqref{series1} into \eqref{P21-ori}, we readily verify \eqref{subconclusion1}.

Following the same approach, we could derive step by step a counterpart of \eqref{P21-ori} as
\begin{align}
\mathbb{P}\left(\frac{\|\va\|^2}{n}\leq 1-\varepsilon\right)
=&\mathbb{P}\left(\myexp{-\lambda\frac{\|\va\|^2}{n}}\ge \myexp{-\lambda\left(1-\varepsilon\right)}\right)\nonumber\\
\le&\frac{\mathbb{E}\left(\myexp{-\lambda\frac{\|\va\|^2}{n}}\right)}{\myexp{-\lambda\left(1-\varepsilon\right)}}\nonumber\\
\le&\left(1+2\frac{\lambda}{n}\right)^{-\frac{n}{2}}\myexp{\lambda\left(1-\varepsilon\right)}\nonumber\\
=&\myexp{\lambda\left(1-\varepsilon\right)-\frac{n}{2}\ln\left(1+2\frac{\lambda}{n}\right)}\nonumber\\
\le& \myexp{\frac{n}{2}\left(\varepsilon+\ln\left(1-\varepsilon\right)\right)}. \label{P22}
\end{align}
Similarly bounding the natural logarithm in \eqref{P22} by truncating the Taylor's series, we have
\begin{align}\label{series2}
\frac{n}{2}\left(\varepsilon+\ln\left(1-\varepsilon\right)\right)
=&\frac{n}{2}\left(\varepsilon+\left(-\varepsilon-\frac{1}{2}\varepsilon^2-\sum_{k=3}^{\infty}\frac{1}{k}\varepsilon^k\right)\right)\nonumber\\
<&-\frac{n\varepsilon^2}{4}\nonumber\\
<&-n\left(\frac{\varepsilon^2}{4}-\frac{\varepsilon^3}{6}\right).
\end{align}
Inserting \eqref{series2} into \eqref{P22}, we verify \eqref{subconclusion2} and close the proof.

\subsection{Proof of Lemma \ref{lem3}}
\label{proof-lem3}

The following lemma is required in the proof.

\begin{lem}\label{lem0}
(\cite{Davidson2001Local})
Let ${\bf A}$ be an $n\times p$ matrix whose entries are independent standard Gaussian random variables.
Then for every $\varepsilon>0$, one has
\begin{align*}
\mathbb{P}\left(s_{\min}\left({\bf A}\right) \ge \sqrt{p}-\sqrt{n}-\varepsilon\right) &\ge 1-\myexp{-\frac{\varepsilon^2}2},\\
\mathbb{P}\left(s_{\max}\left({\bf A}\right)\le \sqrt{p}+\sqrt{n}+\varepsilon\right) &\ge 1-\myexp{-\frac{\varepsilon^2}2},
\end{align*}
where $s_{\min}(\cdot)$ and $s_{\max}(\cdot)$ denote, respectively, the smallest and the largest singular values of a matrix.
\end{lem}

For given set $\mathcal{S}\subset\{1,\cdots,m\}$,
define ${\bf A}_{\mathcal{S}}$ as a matrix in $\mathbb{R}^{n\times p}$ whose columns are $\va_k, k\in\mathcal{S}$.
Let
\begin{align}\label{A-eq}
{\bm\Sigma}_{\mathcal{S}} := \frac{1}{p}{\bf A}_{\mathcal{S}}{\bf A}_{\mathcal{S}}^{\rm T}=\frac{1}{p}\sum_{k\in\mathcal{S}}\va_k\va_k^{\rm T},
\end{align}
and the largest eigenvalue of ${\bm\Sigma}_{\mathcal{S}}$ is calculated by
$$
\lambda_{\max}\left({\bm\Sigma}_{\mathcal{S}}\right) =\frac{1}{p}\left(s_{\max}\left({\bf A}_{\mathcal{S}}\right)\right)^2.
$$
Then according to Lemma \ref{lem0}, we have
\begin{align}\label{eq-temp2}
\mathbb{P}\left(\lambda_{\max}\left({\bm\Sigma}_{\mathcal{S}}\right)\leq \left(1+\frac{1}{\sqrt{p/n}}+\varepsilon_1\right)^2\right)\ge1-\myexp{-\frac{p\varepsilon_1^2}{2}}.
\end{align}
If $\mathcal{S}$ may be any set in $\{1,\cdots, m\}$ with cardinality $p$,
the probability in the RHS of \eqref{eq-temp2} needs to satisfy
$$
P \ge 1-C_{m}^{p}\myexp{-\frac{p\varepsilon_1^2}{2}}
$$
and we prove \eqref{lem-eq3}.

Following the same approach, we represent the smallest eigenvalue of ${\bm\Sigma}_{\mathcal{S}}$ as
$$
\lambda_{\min}\left({\bm\Sigma}_{\mathcal{S}}\right)=\frac{1}{p}\left(s_{\min}\left({\bf A}_{\mathcal{S}}\right)\right)^2
$$
and obtain for $p>n$
\begin{align}\label{eq-temp3}
\mathbb{P}\left(\lambda_{\min}\left({\bm\Sigma}_{\mathcal{S}}\right)\ge \left(1-\frac{1}{\sqrt{p/n}}-\varepsilon_2\right)^2\right)\ge1-\myexp{-\frac{p\varepsilon_2^2}{2}},
\end{align}
by using Lemma \ref{lem0}.
Inserting the combination number again to make it hold for arbitrary $\mathcal{S}$,
we arrive at \eqref{lem-eq2} and complete the proof.


\subsection{Calculate the upper bound $\beta_0$ in \eqref{Eve_pre_final}}
\label{appendix-beta0}

We introduce a Lemma first.

\begin{lem}\label{lem4}
Assume that $a_1,\cdots,a_m$ are independent standard Gaussian random variables and $\left\{a_{(i)}^2\right\}$ are the order statistics of $\{a_{i}^2\}$.
Then for $0<t\le 1$, we have
\begin{align}\label{lim1}
\lim_{m\rightarrow\infty}\frac{1}{tm}\sum_{i=1}^{tm} a_{(i)}^2 = 1 - \frac{1}{t}\frac{2\tau}{\sqrt{2\pi}}\myexp{-\frac{\tau^2}{2}},
\end{align}
where $\tau$ is determined by
$$
1-2Q(\tau) = t,
$$
and $Q\left(\cdot\right)$ is the tail probability of the standard Gaussian distribution.
\end{lem}
\begin{proof}
The proof is postponed to Appendix \ref{proof-lem4}.
\end{proof}

On one hand, according to \eqref{cond-S},
we have
\begin{align}\label{eq-step4-right}
\frac{1}{\beta n}\sum_{k \in \mathcal{S}}y_k^2 \le &\frac{1}{\beta n}\sum_{k \in \mathcal{S}}\left(\va_k^{\rm T}\ve_{t-1}\right)^2\nonumber\\
\leq &\lambda_{\max}\left({\bm\Sigma}_{\mathcal{S}} \right)\|\ve_{t-1}\|^2,
\end{align}
where ${\bm \Sigma}_{\mathcal{S}}$ is defined in \eqref{def-Sigma}.
Let $n$ go to infinity with $\alpha$ fixed,
by applying Lemma \ref{lem3} on \eqref{eq-step4-right}, we have
\begin{equation}\label{eq-step4-combined1}
\frac{1}{\beta n}\sum_{k \in \mathcal{S}}y_k^2\le \left(1 + \frac{1}{\sqrt{\beta}} + \sqrt{2\ln \frac{{\rm e}\alpha}{\beta}}\right)^2\|\ve_{t-1}\|^2
\end{equation}
hold with probability $1$.

On the other hand, we write
\begin{align}
\frac{1}{\beta n}\sum_{k \in \mathcal{S}}y_k^2
= &\|{\bf x}^*\|^2\frac{1}{\beta n}\sum_{k \in \mathcal{S}}(a'_k)^2\label{eq-step4-left-pre}\\
\ge &\|{\bf x}^*\|^2\frac{1}{\beta n}\sum_{k=1}^{\beta n}\left(a'_{(k)}\right)^2,\label{eq-step4-left}
\end{align}
where \eqref{eq-step4-left-pre} comes from the fact that
$$
	y_k^2 = |\va_k^{\rm T}\vx^*|^2 = \|\vx^*\|^2 (a'_k)^2,
$$
and $a'_k$ is a standard Gaussian random variable determined only by $\va_k$,
and $\{a'_{\left(k\right)}\}$ is the order statistics of $\{a'_k\}$, $k=1,\cdots,m$.
When $n$ goes to infinity with $\alpha$ fixed, applying Lemma \ref{lem4} on \eqref{eq-step4-left}, we have
\begin{equation}\label{eq-step4-combined2}
\left(1 - \frac{\alpha}{\beta}\frac{2\tau}{\sqrt{2\pi}}\myexp{-\frac{\tau^2}{2}}\right)\|\vx^*\|^2\le \frac{1}{\beta n}\sum_{k \in \mathcal{S}}y_k^2,
\end{equation}
where
\begin{equation}\label{eq-def-tau}
1-2Q\left(\tau\right)=\frac{\beta}{\alpha}.
\end{equation}
Combining \eqref{eq-step4-combined1} and \eqref{eq-step4-combined2} together, we finally get
\begin{align}\label{eq-step4-combined}
\left(1 - \frac{\alpha}{\beta}\frac{2\tau}{\sqrt{2\pi}}\myexp{-\frac{\tau^2}{2}}\right)\|\vx^*\|^2\le
 \left(1 + \frac{1}{\sqrt{\beta}} + \sqrt{2\ln \frac{{\rm e}\alpha}{\beta}}\right)^2\|\ve_{t-1}\|^2.
\end{align}
The upper bound $\beta_0$ is given as the maximum of $\beta$ satisfying \eqref{eq-step4-combined}.
Considering that the LHS and the RHS of \eqref{eq-step4-combined} are, respectively, an increasing function and a decreasing function of $\beta$,
$\beta_0$ can be determined by solving the equation of \eqref{eq-step4-combined}.

\subsection{Proof of \eqref{probability}}
\label{proof-C1p}

In this subsection we will reshape $P(\beta_0)$ in \eqref{P_pre_final}, where $\beta$ is replaced by its upper bound of $\beta_0$, into \eqref{probability} and verify $C_1>0$.

For the multipliers of the exponential items in $P(\beta_0)$ in \eqref{P_pre_final},
we bound the binomial coefficients as
\begin{align*}
{\rm C}_{m}^{\beta_0 n}=&{\rm C}_{m}^{m-\beta_0 n}
\le\left(\frac{{\rm e}\alpha}{\beta_0}\right)^{\beta_0 n}.
\end{align*}
Then \eqref{P_pre_final} becomes
\begin{align}
P(\beta_0) \ge& 1-2m\myexp{-n\left(\frac{{\hat\varepsilon}_1^2}{4}-\frac{{\hat\varepsilon}_1^3}{6}\right)}
-\myexp{-\frac{\beta_0 n{\hat\varepsilon}_2^2}{2}}
-\myexp{-\frac{\left(\alpha -\beta_0 \right)n{\hat\varepsilon}_3^2}{2}}\nonumber\\
=&1-\myexp{-n\left(\frac{{\hat\varepsilon}_1^2}{4}-\frac{{\hat\varepsilon}_1^3}{6}-\frac{\ln(2m)}{n}\right)}
-\myexp{- \frac{n\beta_0{\hat\varepsilon}_2^2}{2}}
-\myexp{-\frac{n\left(\alpha -\beta_0 \right){\hat\varepsilon}_3^2}{2}}\label{P_temp1},
\end{align}
where ${\hat\varepsilon}_i, i=1,2,3$ are defined in \eqref{def-hat-e-1}, \eqref{def-hat-e-2}, and \eqref{def-hat-e-3}.
Let
\begin{align*}
C_{11}&:=\frac{{\hat\varepsilon}_1^2}{4}-\frac{{\hat\varepsilon}_1^3}{6}-\frac{\ln(6m)}{n},\\
C_{12}&:=\frac{\beta_0{\hat\varepsilon}_2^2}{2}-\frac{\ln3}{n},\\
C_{13}&:=\frac{\left(\alpha -\beta_0 \right){\hat\varepsilon}_3^2}{2}-\frac{\ln3}{n},
\end{align*}
then there exists $n_{C_1}$, when $n>n_{C_1}$, we can choose $0<\hat{\varepsilon}_1<1$, such that
\begin{align*}
C_{11}&>\frac{{\hat\varepsilon}_1^2}{8}-\frac{{\hat\varepsilon}_1^3}{12}=:\hat{C}_{11}>0,\\
C_{12}&>0,\\
C_{13}&>0.
\end{align*}
Let
\begin{align*}
C_1:=\min\left(\hat{C}_{11},C_{12},C_{13}\right) > 0,
\end{align*}
then the probability \label{P_temp1} becomes \eqref{probability}
and we complete the verification.

\subsection{Proof of Lemma \ref{lem4}}
\label{proof-lem4}

Without loss of generality, assume that $a_{(tm)}^2 < M$ when $m$ is large enough.
Denote $x_k = \frac{k}{N}M$, $k=0,\cdots,N$.
Then
$$
0 = x_0 < x_1 < \ldots < x_{N-1} < x_{N} = M,
$$
is a partition of the interval $[0, M]$.
For any $k = 1, \ldots, N$, the ratio of the number $\{a_i^2\}$ falling into the small interval $[x_{k-1}, x_k)$ to the total number $m$, denoted by $p_k^m$, can be expressed by
$$
p_k^{m} = \frac{1}{m}\sum_{i = 1}^m \mathbb{I}(a_i^2 \in [x_{k-1}, x_k)).
$$
Since $\mathbb{I}(a_i^2 \in [x_{k-1}, x_k)), i = 1, \ldots, m$ are \emph{i.i.d.} random variables,
according to the Law of Large Numbers, we have
\begin{align*}
\lim_{m \to \infty} p_k^{m} &= \mathbb{E} \mathbb{I}(a_i^2 \in [x_{k-1}, x_k)) \\
&= \mathbb{P}(a_i^2 \in [x_{k-1}, x_k)) =: p_k.
\end{align*}
Let $K$ satisfy
$$
\sum_{k = 1}^{K} p_k^m \le t \le \sum_{k = 1}^{K+1} p_k^m.
$$
Then we have
$$
\frac{1}{tm}\sum_{i=1}^{tm} a_{(i)}^2 \ge \frac{1}{t}\sum_{k = 1}^{K} x_{k-1}p_k^m,
$$
and
$$
\frac{1}{tm}\sum_{i=1}^{tm} a_{(i)}^2 \le \frac{1}{t}\sum_{k = 1}^{K+1} x_kp_k^m.
$$
Let $m$ go to infinity, then
\begin{align}\label{temp-lem4}
\frac{1}{t}\sum_{k = 1}^{K} x_{k-1}p_k \le \lim_{m\rightarrow\infty}\frac{1}{tm}\sum_{i=1}^{tm} a_{(i)}^2 \le \frac{1}{t}\sum_{k = 1}^{K+1} x_kp_k.
\end{align}
Let $N$ go to infinity, according to the definition of Riemann integral, both the upper and the lower bound in \eqref{temp-lem4} tends to
\begin{align}\label{temp2-lem4}
\frac{1}{t}\int_{-\tau}^{\tau}x^2\phi\left(x\right){\rm d}x,
\end{align}
where
$$
\int_{-\tau}^{\tau}\phi\left(x\right)=t
$$
and $\phi\left(x\right)$ is the probability density function of Gaussian variable.
By calculating this integral in \eqref{temp2-lem4}, one can find that it is exactly the limitation in the LHS of \eqref{lim1}, then we conclude the proof.


\bibliographystyle{IEEEtran}
\bibliography{IEEEabrv,mybibfile}

\begin{thebibliography}{10}
\providecommand{\url}[1]{#1}
\csname url@samestyle\endcsname
\providecommand{\newblock}{\relax}
\providecommand{\bibinfo}[2]{#2}
\providecommand{\BIBentrySTDinterwordspacing}{\spaceskip=0pt\relax}
\providecommand{\BIBentryALTinterwordstretchfactor}{4}
\providecommand{\BIBentryALTinterwordspacing}{\spaceskip=\fontdimen2\font plus
\BIBentryALTinterwordstretchfactor\fontdimen3\font minus
  \fontdimen4\font\relax}
\providecommand{\BIBforeignlanguage}[2]{{%
\expandafter\ifx\csname l@#1\endcsname\relax
\typeout{** WARNING: IEEEtran.bst: No hyphenation pattern has been}%
\typeout{** loaded for the language `#1'. Using the pattern for}%
\typeout{** the default language instead.}%
\else
\language=\csname l@#1\endcsname
\fi
#2}}
\providecommand{\BIBdecl}{\relax}
\BIBdecl

\bibitem{balan2006signal}
R.~Balan, P.~Casazza, and D.~Edidin, ``On signal reconstruction without
  phase,'' \emph{Applied and Computational Harmonic Analysis}, vol.~20, no.~3,
  pp. 345--356, 2006.

\bibitem{conca2015algebraic}
A.~Conca, D.~Edidin, M.~Hering, and C.~Vinzant, ``An algebraic characterization
  of injectivity in phase retrieval,'' \emph{Applied and Computational Harmonic
  Analysis}, vol.~38, no.~2, pp. 346--356, 2015.

\bibitem{harrison1993phase}
R.~W. Harrison, ``Phase problem in crystallography,'' \emph{JOSA a}, vol.~10,
  no.~5, pp. 1046--1055, 1993.

\bibitem{miao2008extending}
J.~Miao, T.~Ishikawa, Q.~Shen, and T.~Earnest, ``Extending x-ray
  crystallography to allow the imaging of noncrystalline materials, cells, and
  single protein complexes,'' \emph{Annu. Rev. Phys. Chem.}, vol.~59, pp.
  387--410, 2008.

\bibitem{fienup1987phase}
C.~Fienup and J.~Dainty, ``Phase retrieval and image reconstruction for
  astronomy,'' \emph{Image Recovery: Theory and Application}, pp. 231--275,
  1987.

\bibitem{bunk2007diffractive}
O.~Bunk, A.~Diaz, F.~Pfeiffer, C.~David, B.~Schmitt, D.~K. Satapathy, and J.~F.
  van~der Veen, ``Diffractive imaging for periodic samples: retrieving
  one-dimensional concentration profiles across microfluidic channels,''
  \emph{Acta Crystallographica Section A: Foundations of Crystallography},
  vol.~63, no.~4, pp. 306--314, 2007.

\bibitem{walther1963question}
A.~Walther, ``The question of phase retrieval in optics,'' \emph{Journal of
  Modern Optics}, vol.~10, no.~1, pp. 41--49, 1963.

\bibitem{shechtman2015phase}
Y.~Shechtman, Y.~C. Eldar, O.~Cohen, H.~N. Chapman, J.~Miao, and M.~Segev,
  ``Phase retrieval with application to optical imaging: a contemporary
  overview,'' \emph{IEEE signal processing magazine}, vol.~32, no.~3, pp.
  87--109, 2015.

\bibitem{gerchberg1972practical}
R.~W. Gerchberg, ``A practical algorithm for the determination of the phase
  from image and diffraction plane pictures,'' \emph{Optik}, vol.~35, pp.
  237--246, 1972.

\bibitem{fienup1982phase}
J.~R. Fienup, ``Phase retrieval algorithms: a comparison,'' \emph{Applied
  optics}, vol.~21, no.~15, pp. 2758--2769, 1982.

\bibitem{candes2013phaselift}
E.~J. Candes, T.~Strohmer, and V.~Voroninski, ``Phaselift: Exact and stable
  signal recovery from magnitude measurements via convex programming,''
  \emph{Communications on Pure and Applied Mathematics}, vol.~66, no.~8, pp.
  1241--1274, 2013.

\bibitem{candes2014solving}
E.~J. Cand{\`e}s and X.~Li, ``Solving quadratic equations via phaselift when
  there are about as many equations as unknowns,'' \emph{Foundations of
  Computational Mathematics}, vol.~14, no.~5, pp. 1017--1026, 2014.

\bibitem{waldspurger2015phase}
I.~Waldspurger, A.~d’Aspremont, and S.~Mallat, ``Phase recovery, maxcut and
  complex semidefinite programming,'' \emph{Mathematical Programming}, vol.
  149, no. 1-2, pp. 47--81, 2015.

\bibitem{netrapalli2013phase}
P.~Netrapalli, P.~Jain, and S.~Sanghavi, ``Phase retrieval using alternating
  minimization,'' in \emph{Advances in Neural Information Processing Systems},
  2013, pp. 2796--2804.

\bibitem{Wei2015Solving}
K.~Wei, ``Solving systems of phaseless equations via kaczmarz methods: A proof
  of concept study,'' \emph{Inverse Problems}, vol.~31, no.~12, p. 125008,
  2015.

\bibitem{li2015phase}
G.~Li, Y.~Gu, and Y.~M. Lu, ``Phase retrieval using iterative projections:
  Dynamics in the large systems limit,'' in \emph{Communication, Control, and
  Computing (Allerton), 2015 53rd Annual Allerton Conference on}.\hskip 1em
  plus 0.5em minus 0.4em\relax IEEE, 2015, pp. 1114--1118.

\bibitem{candes2015phase}
E.~J. Candes, X.~Li, and M.~Soltanolkotabi, ``Phase retrieval via wirtinger
  flow: Theory and algorithms,'' \emph{IEEE Transactions on Information
  Theory}, vol.~61, no.~4, pp. 1985--2007, 2015.

\bibitem{chen2015solving}
Y.~Chen and E.~Candes, ``Solving random quadratic systems of equations is
  nearly as easy as solving linear systems,'' in \emph{Advances in Neural
  Information Processing Systems}, 2015, pp. 739--747.

\bibitem{tan2017phase}
Y.~S. Tan and R.~Vershynin, ``Phase retrieval via randomized kaczmarz:
  Theoretical guarantees,'' \emph{arXiv preprint arXiv:1706.09993}, 2017.

\bibitem{jeong2017convergence}
H.~Jeong and C.~S. G{\"u}nt{\"u}rk, ``Convergence of the randomized kaczmarz
  method for phase retrieval,'' \emph{arXiv preprint arXiv:1706.10291}, 2017.

\bibitem{kaczmarz1937angenaherte}
S.~Kaczmarz, ``Angenaherte auflosung von systemen linearer gleichungen,''
  \emph{Bull. Int. Acad. Sci. Pologne, A}, vol.~35, pp. 355--357, 1937.

\bibitem{strohmer2006randomized}
T.~Strohmer and R.~Vershynin, ``A randomized solver for linear systems with
  exponential convergence,'' \emph{Lecture Notes in Computer Science}, vol.
  4110, p. 499, 2006.

\bibitem{lu2017phase}
Y.~M. Lu and G.~Li, ``Phase transitions of spectral initialization for
  high-dimensional nonconvex estimation,'' \emph{arXiv preprint
  arXiv:1702.06435}, 2017.

\bibitem{Davidson2001Local}
K.~R. Davidson and S.~J. Szarek, ``Local operator theory, random matrices and
  banach spaces,'' \emph{“handbook in Banach Spaces” Vol}, pp. 317--366,
  2001.

\end{thebibliography}
\vfill\pagebreak

\end{document}